\documentclass[11pt]{article}
\usepackage{amsfonts}
\usepackage{amsmath}

\newcommand{\be}{\begin{equation}}
\newcommand{\ee}{\end{equation}}
\newcommand{\bea}{\begin{eqnarray}}
\newcommand{\eea}{\end{eqnarray}}
\newcommand{\beas}{\begin{eqnarray*}}
\newcommand{\eeas}{\end{eqnarray*}}

\newtheorem{theorem}{Theorem}[section]
\newtheorem{definition}[theorem]{Definition}

\newtheorem{remark}[theorem]{Remark}
\newtheorem{example}[theorem]{Example}
\newtheorem{examples}[theorem]{Examples}
\newtheorem{foo}[theorem]{Remarks}

\newenvironment{proof}{\addvspace{\medskipamount}\par\noindent{\it Proof}.}
{\unskip\nobreak\hfill$\Box$\par\addvspace{\medskipamount}}

\newcommand{\brak}[1]{\left(#1\right)}    
\newcommand{\crl}[1]{\left\{#1\right\}}   
\newcommand{\edg}[1]{\left[#1\right]}     

\newcommand{\p}{\mathbb{P}}
\newcommand{\q}{\mathbb{Q}}

\newcommand{\E}[1]{{\mathbb{E}}\left[#1\right]}

\newcommand{\EQ}[1]{{\mathbb{E}}_{\mathbb{Q}}\left[#1\right]}

\newcommand{\abs}[1]{\left|#1\right|}     

\DeclareMathOperator{\argmax}{arg\,max}

\begin{document}

\title{Optimal consumption and investment\\
in incomplete markets with general constraints\thanks{The paper was written during periods when
the authors could work together at the Universit\'e Rennes 1 or Princeton University.
The hospitality of both institutions is greatly appreciated. We also thank an anonymous
referee for valuable comments.}}
\author{Patrick Cheridito\thanks{Partially supported by NSF Grant DMS-0642361}\\
Princeton University\\ Princeton,
NJ 08544, USA \and Ying Hu\thanks{Partially supported by
funds from the Marie Curie ITN Grant, ``Controlled Systems'', GA no.213841/2008.}
\\ IRMAR,
Universit\'e Rennes 1\\ 35042 Rennes Cedex,
France.}

\date{}

\maketitle

{\bf Abstract.} We study an optimal consumption and investment problem
in a possibly incomplete market with general, not necessarily convex,
stochastic constraints. We provide explicit solutions
for investors with exponential, logarithmic as well as power utility and
show that they are unique if the constraints are convex.
Our approach is based on martingale methods that rely on
results on the existence and uniqueness of solutions to
BSDEs with drivers of quadratic growth.

\section{Introduction}

We consider an investor receiving stochastic income who can invest in a
financial market. The question is how to optimally consume and invest if
utility is derived from intermediate consumption and the level of
remaining wealth at some final time $T$.
More specifically, we assume our investor receives income at rate $e_t$ and
a lump sum payment $E$ at the final time. The investor chooses a rate of consumption $c_t$ and
an investment policy so as to maximize the expectation
$$
\E{\int_0^T \alpha e^{-\beta t} u(c_t) dt + e^{-\beta T} u(X_T + E)},
$$
where $\alpha$ and $\beta$ are constants,
$u : \mathbb{R} \to \mathbb{R} \cup \crl{- \infty}$ is a concave utility function
and $X_T$ is his/her wealth immediately before receiving
the lump sum payment $E$. There exists an extensive literature on problems of this form;
see for instance, Karatzas and Shreve \cite{KS} for an overview.

The novelty of this paper is that we put general, not necessarily convex,
stochastic constraints on consumption and investment. We provide
explicit solutions for investors with exponential, logarithmic and power utility
in terms of solutions to BSDEs with drivers of quadratic growth.
Our approach is based on an extension of the arguments of Hu et al. \cite{HIM}, where
investment problems without intermediate consumption are studied.
To every admissible strategy we associate a utility process, which we show to always be
a supermartingale and a martingale if and only if the strategy is optimal. This method
relies on results from Kobylanski \cite{Kobylanski} and Morlais \cite{Morlais}
on the existence and properties of solutions to BSDEs with drivers of quadratic growth.
We formulate constraints on consumption and investment in terms of subsets of predictable processes and use
conditional analysis results from Cheridito et al. \cite{CKV} to obtain the existence of
optimal strategies. For related results in a slightly different setup, see Nutz \cite{Nutz},
where dynamic programming is used to derive the Bellman equation for power utility
maximization in general semimartingale
models with stochastic constraints on investment.

The structure of the paper is as follows: Section \ref{sec:model} introduces
the model. In Section \ref{sec:cara} we discuss the case of constant absolute risk aversion,
corresponding to exponential utility functions. Section \ref{sec:crra} treats the case of
constant relative risk aversion, which is covered by logarithmic and power utility functions.
The specification of the constraints and the definition of admissible strategies will
be slightly different from case to case. Section \ref{sec:conclusion} concludes with a discussion of the
assumptions and potential generalizations.

\setcounter{equation}{0}
\section{The model}
\label{sec:model}

Let $T \in \mathbb{R}_+$ be a finite time horizon and $(W_t)_{0 \le t \le T}$
an $n$-dimensional Brownian motion on a probability
space $(\Omega, {\cal F}, \p)$. Denote by $({\cal F}_t)$ the
augmented filtration generated by $(W_t)$. We consider a financial market
consisting of a money market and $m \le n$ stocks. Money can be lent to and
borrowed from the money market at a constant interest rate
$r \ge 0$ and the stock prices follow the dynamics
$$
\frac{dS^i_t}{S^i_t} = \mu^i_t dt + \sigma^i_t dW_t, \quad S^i_0 > 0, \quad i =1 , \dots, m,
$$
for bounded predictable processes $\mu^i_t$ and $\sigma^i_t$ taking values in
$\mathbb{R}$ and $\mathbb{R}^{1 \times n}$, respectively. If $m < n$, the
stocks do not span all uncertainty and the market is incomplete even if there
are no constraints.

Consider an investor with initial wealth $x \in \mathbb{R}$ receiving income at
a predictable rate $e_t$ and an ${\cal F}_T$-measurable lump sum payment
$E$ at time $T$ who can consume at intermediate times and invest in the
financial market. If the investor consumes at a predictable rate $c_t$ and
invests according to a predictable trading strategy $\pi_t$ taking
values in $\mathbb{R}^{1 \times m}$, where $\pi^i_t$ is
the amount of money invested in stock $i$ at time $t$,
his/her wealth evolves like
$$
X_t = x+ \int_0^t \brak{X_s - \sum_{i=1}^m \pi^i_s} r ds +
\sum_{i=1}^m \int_0^t \frac{\pi^i_s}{S^i_s} dS^i_s + \int_0^t (e_s-c_s) ds.
$$
Denote by $\sigma_t$ the matrix with rows $\sigma^i_t$, $i = 1, \dots, m$.
Assume that $\sigma \sigma^T$ is invertible $\nu \otimes \p$-almost everywhere, where
$\nu$ is the Lebesgue measure on $[0,T]$, and the process
$$
\theta = \sigma^T (\sigma \sigma^T)^{-1} (\mu - r1)
$$
is bounded. Then for $p = \pi \sigma$, one can write
\be \label{XSDE}
X^{(c,p)}_t = x + \int_0^t X^{(c,p)}_s r ds +
\int_0^t p_t [dW_t + \theta_t dt] + \int_0^t (e_s-c_s) ds.
\ee
Note that if
$$
\int_0^T (|e_t| + |c_t| + |p_t|^2) dt < \infty \quad \p \mbox{-almost surely,}
$$
where $|.|$ denotes the Euclidean norm on $\mathbb{R}^{1 \times n}$, then
$$
\int_0^t p_t [dW_t + \theta_t dt] + \int_0^t (e_s-c_s) ds
$$
is a continuous stochastic process, and it follows that equation \eqref{XSDE} has a unique continuous
solution $X^{(c,p)}$; see for instance Remark A.2 in Cheridito et al. \cite{CHM}.

We assume our agent chooses $c$ and $\pi$ so as to maximize
\be \label{opt}
\E{\int_0^T \alpha e^{-\beta t} u(c_t) dt + e^{-\beta T} u \brak{X^{(c,p)}_T + E}}
\ee
for given constants $\alpha > 0$, $\beta \in \mathbb{R}$ and
a concave function $u : \mathbb{R} \to \mathbb{R} \cup \crl{-\infty}$.
The specific cases we will discuss are:\\
\hspace*{5mm} $\bullet$ $u(x) = - \exp(-\gamma x)$ for $\gamma > 0$\\
\hspace*{5mm} $\bullet$ $u(x) = \log(x)$\\
\hspace*{5mm} $\bullet$ $u(x) = x^{\gamma}/\gamma$ for $\gamma \in (-\infty,0) \cup (0,1)$.\\
As usual, for $\gamma > 0$, we understand $x^{\gamma}/\gamma$ to be $-\infty$ on $(-\infty,0)$
while $\log(x)$ and $x^{\gamma}/\gamma$ for $\gamma < 0$ are meant to be $-\infty$ on
$(-\infty ,0]$.

To formulate consumption and investment constraints we introduce non-empty subsets
$C \subset {\cal P}$ and $Q \subset {\cal P}^{1 \times m}$, where ${\cal P}$ denotes the
set of all real-valued predictable processes $(c_t)_{0 \le t \le T}$ and
${\cal P}^{1 \times m}$ the set of all predictable processes $(\pi_t)_{0 \le t \le T}$ with values in
$\mathbb{R}^{1 \times m}$. In Section \ref{sec:cara} we do not put restrictions on
consumption and just require the investment strategy $\pi$ to belong to $Q$. In Section \ref{sec:crra}
consumption and investment will be of the form $c = \tilde{c} X$
and $\pi = \tilde{\pi} X$, respectively, and we will require $\tilde{c}$ to be in $C$ and
$\tilde{\pi}$ in $Q$.

Note that the expected value \eqref{opt} does not change if $(c,p)$ is
replaced by a pair $(c',p')$ which is equal $\nu \otimes \p$-a.e.
So we identify predictable processes that agree
$\nu \otimes \p$-a.e. and use the following concepts from Cheridito et al. \cite{CKV}:
We call a subset $A$ of ${\cal P}^{1 \times k}$ {\bf sequentially closed} if
it contains every process $a$ that is the $\nu \otimes \p$-a.e. limit
of a sequence $(a^n)_{n \ge 1}$ of processes in $A$. We call it
${\cal P}$-{\bf stable} if it contains $1_B a + 1_{B^c} a'$ for all $a,a' \in A$ and
every predictable set $B \subset [0,T] \times \Omega$. We say $A$ is
${\cal P}$-{\bf convex} if it contains $\lambda a + (1-\lambda) a'$
for all $a,a' \in A$ and every process $\lambda \in {\cal P}$ with values in $[0,1]$.
In the whole paper we work with the following\\[3mm]
\hspace*{7mm}{\bf Standing assumption} \quad $C$ and $Q$ are sequentially closed and ${\cal P}$-stable.\\[3mm]
This will allow us to show existence of optimal strategies. If, in addition,
$C$ and $Q$ are ${\cal P}$-convex, the optimal strategies will be unique.
Note that $P = \crl{\pi \sigma : \pi \in Q}$ is a ${\cal P}$-stable subset
of ${\cal P}^{1 \times n}$, which, since we
assumed $\sigma \sigma^T$ to be invertible for $\nu \otimes \p$-almost all $(t,\omega)$,
is ${\cal P}$-convex if and only if $Q$ is. Moreover, it follows from \cite{CKV}
that $P$ is  sequentially closed.

For a process $q$ in ${\cal P}^{1 \times n}$, we denote by ${\rm dist}(q,P)$ the
predictable process
$$
{\rm dist}(q,P) := \mathop{\rm ess\,inf}_{p \in P} |q-p|,
$$
where ess\,inf denotes the greatest lower bound with respect to the $\nu \otimes \p$-a.e. order.
It is shown in \cite{CKV} that there exists a
process $p \in P$ satisfying $|q-p| = {\rm dist}(q,P)$ and that it is unique
(up to $\nu \otimes \p$-a.e. equality) if $P$ is ${\cal P}$-convex.
We denote the set of all these processes by $\Pi_P(q)$.

By ${\cal P}^{1 \times n}_{\rm BMO}$ we denote the processes
$Z \in {\cal P}^{1 \times n}$
for which there exists a constant $D \ge 0$ such that
$$
\E{\int_{\tau}^T |Z_t|^2 dt \mid {\cal F}_{\tau}} \le D \quad
\mbox{for all stopping times } \tau \le T.
$$
For every $Z \in {\cal P}^{1 \times n}_{\rm BMO}$,
$\int_0^. Z_s dW_s$ is a BMO-martingale and
${\cal E}(Z \cdot W)_t$, $0 \le t \le T$, a positive martingale.
Moreover, if $Z,V$ belong to ${\cal P}^{1 \times n}_{\rm BMO}$, then $Z$ is also
in ${\cal P}^{1 \times n}_{\rm BMO}$ with respect to the Girsanov transformed measure
$$
\q = {\cal E}(V \cdot W)_T \cdot \p;
$$
see for instance, Kazamaki \cite{Kaz}.

\setcounter{equation}{0}
\section{CARA or exponential utility}
\label{sec:cara}

We first assume that our investor has constant absolute risk aversion
$- u''(x)/u'(x) = \gamma > 0$. Then, up to affine transformations,
the utility function $u$ is of the exponential form
$$
u(x) = - \exp(-\gamma x).
$$
Here we do not constrain consumption, that is, $C = {\cal P}$, and we assume
that the set $P$ of possible investment strategies contains at least one
bounded process $\bar{p}$. Moreover, we assume that the rate of income
$e$ and the final payment $E$ are both bounded.

Define the bounded positive function $h$ on $[0,T]$ by
$$
h(t) = 1/(1 + T -t) \quad \mbox{if} \quad r = 0
$$
and
$$
h(t) = \frac{r}{1 -(1-r) \exp(-r (T-t))}  \quad \mbox{if} \quad r > 0.
$$
Note that in both cases $h$ solves the quadratic ODE
$$
h'(t) = h(t)(h(t) -r), \quad h(T) = 1.
$$

\begin{definition}
If $u(x) = - \exp(-\gamma x)$, an admissible strategy consists of a pair
$(c,p) \in {\cal P} \times P$
such that $\int_0^T (|c_t| + |p_t|^2) dt < \infty$ $\p$-a.s.,
$$
\exp\brak{- \gamma h(t) X^{(c,p)}_t}_{0 \le t \le T} \mbox{ is of class {\rm (D)}} \quad \mbox{and }
\int_0^T \E{e^{- \gamma c_t}} dt < \infty.
$$
\end{definition}

Consider the BSDE
\be \label{BSDEexp}
Y_t = E + \int_t^T f(s,Y_s,Z_s)ds + \int_t^T Z_s dW_s
\ee
with driver
$$
f(t,y,z)
= - \frac{\gamma}{2} {\rm dist}^2_t \brak{z + \frac{1}{\gamma} \theta, h P}
+ z \theta_t + \frac{1}{2 \gamma} |\theta_t|^2
+ h(t) (e_t - y) + \frac{h(t)}{\gamma} \brak{\log \frac{h(t)}{\alpha} - 1}
+ \frac{\beta}{\gamma}.
$$
Since $\theta$, $e$, $E$ and $h$ are bounded and the set $P$ contains a bounded process $\bar{p}$,
there exists a constant $K \in \mathbb{R}_+$ such that
$$|f(t,y,z)| \le K(1+ |y| + |z|^2)$$
and
$$|f(t,y_1,z_1)-f(t,y_2,z_2)|\le K( |y_1-y_2|+(1+|z_1|+|z_2|)|z_1-z_2|).$$
So it follows from Kobylanski \cite{Kobylanski} that equation \eqref{BSDEexp} has a unique solution
$(Y,Z)$ such that $Y$ is bounded and from  Morlais \cite{Morlais} that $Z$
belongs to ${\cal P}^{1 \times n}_{\rm BMO}$.

\begin{theorem} \label{thmexp}
The optimal value of the optimization problem \eqref{opt} for
$u(x) = - \exp(-\gamma x)$ over all admissible strategies is
\be \label{optvalueexp}
- \exp \edg{- \gamma (h(0) x + Y_0)},
\ee
and $(c^*, p^*)$ is an optimal admissible strategy if and only if
\be \label{optstexp}
c^* = h X^{(c^*,p^*)} + Y - \frac{1}{\gamma} \log \frac{h}{\alpha}
\quad \mbox{$\nu \otimes \p$-a.e.} \quad \mbox{and} \quad
p^* \in \Pi_{P} \brak{\frac{Z + \theta/\gamma}{h}}.
\ee
In particular, an optimal admissible strategy exists, and it is unique up to
$\nu \otimes \p$-a.e. equality if the set $P$ is ${\cal P}$-convex.
\end{theorem}

\begin{proof}
For every admissible strategy $(c,p)$, equation \eqref{XSDE} defines a
continuous stochastic process $X^{(c,p)}$. The process
$$
R^{(c,p)}_t = - e^{- \beta t} e^{-\gamma \brak{h(t) X^{(c,p)}_t + Y_t}}
- \int_0^t \alpha e^{- \beta s} e^{- \gamma c_s} ds
$$
satisfies
$$
R^{(c,p)}_0 = - e^{- \gamma (h(0) x + Y_0)}, \quad R^{(c,p)}_T =
- e^{- \beta T} e^{-\gamma \brak{X^{(c,p)}_T + E}}
- \int_0^T \alpha e^{- \beta s} e^{- \gamma c_s} ds
$$
and
$$
dR^{(c,p)}_t = \gamma e^{- \beta t} e^{- \gamma \brak{h(t) X^{(c,p)}_t + Y_t}}
\edg{(h(t) p_t - Z_t) dW_t +  A^{(c,p)}_t dt},
$$
where
\beas
A^{(c,p)}_t = && h(t) p_t \theta_t - \frac{\gamma}{2} |h(t) p_t - Z_t|^2 - f(t,Y_t,Z_t)\\
&& + h(t) (e_t - c_t) - \frac{\alpha}{\gamma} e^{\gamma \brak{h(t) X^{(c,p)}_t + Y_t}} e^{-\gamma c_t}
+ h'(t) X^{(c,p)}_t + h(t) r X^{(c,p)}_t + \frac{\beta}{\gamma}.
\eeas
First note that
\beas
&& h(t) p_t \theta_t - \frac{\gamma}{2} |h(t) p_t - Z_t|^2
= - \frac{\gamma}{2} \abs{h(t) p_t - \brak{Z_t + \frac{1}{\gamma} \theta_t}}^2
+ Z_t \theta_t + \frac{1}{2 \gamma} |\theta_t|^2\\
&\le& - \frac{\gamma}{2} \, {\rm dist}^2_t \, \brak{Z + \frac{1}{\gamma} \theta, h P}
+ Z_t \theta_t + \frac{1}{2 \gamma} |\theta_t|^2,
\eeas
and the inequality becomes a $\nu \otimes \p$-a.e. equality if and only if
$$
p \in \Pi_P \brak{\frac{Z + \theta/\gamma}{h}}.
$$
Furthermore, for fixed $(t,\omega) \in [0,T] \times \Omega$,
$$
z \mapsto - h(t) z - \frac{\alpha}{\gamma} e^{\gamma \brak{h(t) X^{(c,p)}_t + Y_t}} e^{-\gamma z}
$$
is a strictly concave function that is equal to its maximum
$$
\frac{h(t)}{\gamma} \log \frac{h(t)}{\alpha} - h^2(t) X^{(c,p)}_t - h(t) Y_t
- \frac{h(t)}{\gamma}
$$
if and only if
$$
z = h(t) X^{(c,p)}_t + Y_t - \frac{1}{\gamma} \log \frac{h(t)}{\alpha}.
$$
Therefore, one has
\bea
\notag
&& h(t) (e_t - c_t) - \frac{\alpha}{\gamma} e^{\gamma \brak{h(t) X^{(c,p)}_t + Y_t}} e^{-\gamma c_t}
+ h'(t) X^{(c,p)}_t + h(t) r X^{(c,p)}_t + \frac{\beta}{\gamma}\\
\notag
&\le& h(t) e_t + \frac{h(t)}{\gamma} \log \frac{h(t)}{\alpha} - h^2(t) X^{(c,p)}_t - h(t) Y_t
- \frac{h(t)}{\gamma}
+ h'(t) X^{(c,p)}_t + h(t) r X^{(c,p)}_t + \frac{\beta}{\gamma}\\
\label{disX}
&=& h(t) e_t + \frac{h(t)}{\gamma} \log \frac{h(t)}{\alpha} - h(t) Y_t - \frac{h(t)}{\gamma}
+ \frac{\beta}{\gamma},
\eea
where the inequality is attained if and only if
$$
c = h X^{(c,p)} + Y - \frac{1}{\gamma} \log \frac{h}{\alpha}
$$
(note that in \eqref{disX} the $X^{(c,p)}_t$-terms disappear due to our choice of the function $h$).
It follows that for every admissible pair $(c,p)$, $R^{(c,p)}$ is a local supermartingale,
which by our definition of admissible strategies, is of class (D).
Therefore, it is a supermartingale, and one obtains
$$
R^{(c,p)}_0 \ge \E{R^{(c,p)}_T},
$$
where the inequality is strict if the pair $(c,p)$ does not satisfy condition \eqref{optstexp}.
On the other hand, if we can show that each pair $(c^*,p^*)$ satisfying \eqref{optstexp} is admissible
and $R^* = R^{(c^*,p^*)}$ is a martingale, we can conclude that
$$
R^*_0 = \E{R^*_T},
$$
and it follows that $(c^*,p^*)$ is optimal.

But if $(c^*,p^*)$ satisfies \eqref{optstexp}, $c^*$ is continuous in $t$.
In particular, it belongs to ${\cal P}$ and $\int_0^T |c^*_t| dt < \infty$ $\p$-a.s.
Moreover, since $\theta$ as well as $h$ are bounded and $P$ contains a
bounded process $\bar{p}$, there exists a constant $L$ such that
$|p^*| \le L(1 + |Z|)$. It follows that $p^* \in {\cal P}^{1 \times n}_{\rm BMO}$,
and hence, $\int_0^T |p^*_t|^2 dt < \infty$ $\p$-a.s.
Since $A^* := A^{(c^*,p^*)} = 0$, $- R^*$ is a positive local martingale, and one obtains
$$
\E{e^{- \gamma X^*_T}} + \E{\int_0^T e^{- \gamma c^*_t} dt}
\le M \E{- R^*_T} < \infty,
$$
where $M$ is a suitable constant and the inequality $\E{- R^*_T} < \infty$ follows from
Fatou's lemma. By Girsanov's theorem,
$$
W^{\q}_t = W_t + \int_0^t \theta_s ds
$$
is an $n$-dimensional Brownian motion under the measure
$$
\q = {\cal E}(- \theta \cdot W)_T \cdot \p,
$$
and one has
\bea
\notag
&& d(h(t) X^*_t) = h'(t) X^*_t dt + h(t)p^*_t dW_t
+ h(t) [X^*_t r + p^*_t \theta_t + e_t-c^*_t] dt\\
\notag
&=& h'(t) X^*_t dt + h(t)p^*_t dW_t
+ h(t) \edg{X^*_t r + p^*_t \theta_t + e_t- h(t)X^*_t - Y_t +
\frac{1}{\gamma} \log \brak{\frac{h(t)}{\alpha}}} dt\\
\notag
&=& h(t)p^*_t dW_t
+ h(t) \edg{p^*_t \theta_t + e_t - Y_t + \frac{1}{\gamma} \log \brak{\frac{h(t)}{\alpha}}} dt\\
\label{hXbound}
&=& h(t) p^*_t dW^{\q}_t + h(t) \edg{e_t - Y_t
+ \frac{1}{\gamma} \log \brak{\frac{h(t)}{\alpha}}} dt.
\eea
Since $p^*$ belongs to ${\cal P}^{1 \times n}_{\rm BMO}$,
the process $V_t = \int_0^t h(s) p^*_s dW^{\q}_s$ is a BMO-martingale under $\q$, and
it can be seen from \eqref{hXbound} that there exist constants $d_1,d_2$ such that
$$
e^{-\gamma h(t) X^*_t} \le d_1 e^{- \gamma V_t} \quad \mbox{and} \quad
e^{- \gamma V_t} \le d_2 e^{-\gamma h(t) X^*_t} \quad
\mbox{for all } t \in [0,T].
$$
Hence, one obtains for every stopping time $\tau \le T$,
\beas
e^{-\gamma h(\tau) X^*_{\tau}} &\le& d_1 e^{-\gamma V_{\tau}}
\le d_1 \brak{\EQ{e^{- \frac{\gamma}{2} V_T} \mid {\cal F}_{\tau}}}^2\\
&=& d_1 \brak{\E{e^{- \frac{\gamma}{2} V_T} {\cal E}(-\theta \cdot W)_T \mid {\cal F}_{\tau}}}^2
{\cal E}(-\theta \cdot W)_{\tau}^{-2}\\
&\le& d_1 \E{e^{- \gamma V_T} \mid {\cal F}_{\tau}}
\E{{\cal E}(-\theta \cdot W)_T^2 \mid {\cal F}_{\tau}}
{\cal E}(-\theta \cdot W)_{\tau}^{-2}\\
&\le& d_1 d_2 \E{e^{-\gamma X^*_T} \mid {\cal F}_{\tau}}
\E{{\cal E}(-\theta \cdot W)_T^2 \mid {\cal F}_{\tau}}
{\cal E}(-\theta \cdot W)_{\tau}^{-2}.
\eeas
But since $\theta$ is bounded, there exists a constant $d_3$ such that
\beas
&& \E{{\cal E}(-\theta \cdot W)_T^2 \mid {\cal F}_{\tau}}
{\cal E}(-\theta \cdot W)_{\tau}^{-2}\\
&=& \E{\frac{{\cal E}(-2 \theta \cdot W)_T}{{\cal E}(- 2 \theta \cdot W)_{\tau}}
\exp \brak{\int_{\tau}^T |\theta_s|^2 ds} \mid {\cal F}_{\tau}}\\
&\le& d_3 \quad \mbox{for every stopping time } \tau \le T.
\eeas
So one has
$$
e^{-\gamma h(\tau) X^*_{\tau}} \le
d_1 d_2 d_3 \E{e^{- \gamma X^*_T} \mid {\cal F}_{\tau}}
\quad \mbox{for every stopping time } \tau \le T.
$$
This shows that $\exp\brak{- \gamma h(t) X^*_t}_{0 \le t \le T}$ is of class (D).
Therefore, $(c^*,p^*)$ is admissible and $R^*$ a martingale.

It remains to show that a pair $(c^*, p^*)$ satisfying \eqref{optstexp} exists and that it
is unique up to $\nu \otimes \p$-a.e. equality if the set $P$ is ${\cal P}$-convex. It is shown in \cite{CKV}
that a process $p^*$ in $\Pi_{P} \brak{\frac{Z + \theta/\gamma}{h}}$ exists and that it
is unique up to $\nu \otimes \p$-a.e. equality if $P$ is ${\cal P}$-convex.
As we have seen above, every $p^* \in \Pi_{P} \brak{\frac{Z + \theta/\gamma}{h}}$ is also in
${\cal P}^{1 \times n}_{\rm BMO}$. So there exists a unique continuous process $(X_t)$ satisfying
$$
X_t = x + \int_0^t X_s r ds +
\int_0^t p^*_t [dW_t + \theta_t dt]
+ \int_0^t \brak{e_s- h(s) X_s - Y_s + \frac{1}{\gamma} \log \frac{h(s)}{\alpha}} ds.
$$
But $X = X^{(c^*,p^*)}$ for
$$
c^*_t = h(t) X_t + Y_t - \frac{1}{\gamma} \log \frac{h(t)}{\alpha}.
$$
So $(c^*, p^*)$ satisfies condition \eqref{optstexp}, and it is unique
up to $\nu \otimes \p$-a.e. equality if the set $P$ is ${\cal P}$-convex.
\end{proof}

\setcounter{equation}{0}
\section{CRRA utility}
\label{sec:crra}

We now assume that the investor has constant relative risk aversion
$- x u''(x)/u'(x) = \delta > 0$. For $\delta = 1$, this corresponds to
$u(x) = \log(x)$, and for $\delta \neq 1$ to $u(x) = x^{\gamma}/\gamma$,
where $\gamma = 1 - \delta$. We discuss the cases $\delta = 1$ and
$\delta \neq 1$ separately. In both of them we assume $E = 0$.

We here suppose that the initial wealth is strictly positive: $x >0$.
To avoid $-\infty$ utility, the agent must keep the wealth process positive.
Therefore, we can parameterize $e$, $c$ and $\pi$ by
$\tilde{e} = e/X$, $\tilde{c} = c/X$ and $\tilde{\pi} = \pi/X$, respectively.
If one denotes $\tilde{p} = \tilde{\pi} \sigma$,
the corresponding wealth evolves according to
$$
\frac{d X^{(c,p)}_t}{X^{(c,p)}_t} = \tilde{p}_t(dW_t + \theta_t dt)
+ (r + \tilde{e}_t - \tilde{c}_t) dt, \quad X^{(c,p)}_0 = x,
$$
and one can write
\be \label{posX}
X^{(c,p)}_t =
x \, {\cal E} \brak{\tilde{p} \cdot W^{\q}}_t \exp\brak{\int_0^t( r
+ \tilde{e}_s - \tilde{c}_s) ds} > 0,
\ee
where ${\cal E}$ is the stochastic exponential and $W^{\q}_t = W_t +\int_0^t \theta_sds$.

In the whole section we assume that $\tilde{e}$ is bounded and the constraints are of the following form:
$\tilde{c}$ must be in the set $C$ and $\tilde{\pi}$ in $Q$, or equivalently,
$\tilde{p}$ in $P = \crl{\tilde{\pi} \sigma : \tilde{\pi} \in Q}$.
Additionally, $\tilde{c}$ will be required to be positive or non-negative
depending on the specific utility function being used. Also, $\tilde{c}$
and $\tilde{p}$ will have to satisfy suitable integrability conditions.
For all CRRA utility functions $u$ we make the following assumption:
\be \label{asscrra}
\mbox{there exists a pair }
(\bar{c}, \bar{p}) \in C \times P \mbox{ such that }
u(\bar{c}) - \bar{c} \mbox{ and } \bar{p} \mbox{ are bounded.}
\ee
Note that this implies that $u(\bar{c})$ and $\bar{c}$ are both bounded.

\subsection{Logarithmic utility}
\label{subsec:log}

In the case $u(x) = \log(x)$, we introduce the positive function
$$
h(t) = \left\{
\begin{array}{cl}
1+ \alpha (T-t) & \mbox{ if } \beta = 0\\
\alpha/\beta+(1-\alpha/\beta)e^{-\beta (T-t)}
& \mbox{ if } \beta > 0
\end{array}
\right. ,
$$
and notice that
$$h'(t)=\beta h(t)-\alpha \quad \mbox{with} \quad h(T)=1.$$

\begin{definition}
For $u(x) = \log(x)$, an admissible strategy is a pair $(\tilde{c},\tilde{p}) \in C \times P$
satisfying
\be \label{condlog}
\mathbb E\left[\int_0^T|\log( \tilde{c}_t)| dt+\int_0^T \tilde{c}_t dt+\int_0^T
|\tilde{p}_t|^2 dt\right] < \infty.
\ee
\end{definition}
Remember that we understand $\log(x)$ to be $-\infty$ for $x \le 0$. Therefore,
\eqref{condlog} implies $\tilde{c} > 0$ $\nu \otimes \p$-a.e. Let us set
\be \label{maxlog}
\max_{\tilde{c}\in C} \left(\frac{\alpha}{h} \log(\tilde{c})-\tilde{c}\right)
:= \mathop{\rm ess\,sup}_{\tilde{c} \in C}
\left(\frac{\alpha}{h} \log(\tilde{c})-\tilde{c}\right),
\ee
where ess\,sup is the smallest upper bound with respect to $\nu \otimes \p$-a.e.
inequality.  By
\be \label{argmaxlog}
\argmax_{\tilde{c}\in C} \left(\frac{\alpha}{h} \log(\tilde{c})-\tilde{c}\right)
\ee
we denote the set of all processes in $C$ which attain the ess\,sup. It follows
from Cheridito et al. \cite{CKV} that \eqref{argmaxlog} is not
empty and, up to $\nu\otimes\p$-a.e. equality, contains exactly one process if $C$ is ${\cal P}$-convex.
Note that
$$
\frac{\alpha}{h} \log(\bar{c})-\bar{c}
\le \max_{\tilde{c}\in C} \left(\frac{\alpha}{h} \log(\tilde{c})-\tilde{c}\right)
\le \frac{\alpha}{h} \brak{\log \frac{\alpha}{h} -1 },
$$
where $\bar{c}$ is the process of assumption \eqref{asscrra}. It follows that
$\max_{\tilde{c}\in C} \left(\frac{\alpha}{h} \log(\tilde{c})-\tilde{c}\right)$ as well as every
process $\tilde{c} \in \argmax_{\tilde{z}\in C} \left(\frac{\alpha}{h} \log(\tilde{z})-\tilde{z}\right)$
is bounded. In particular, $\log(\tilde{c})$ is bounded for every
$\tilde{c} \in \argmax_{\tilde{z}\in C} \left(\frac{\alpha}{h} \log(\tilde{z})-\tilde{z}\right)$.

Consider the BSDE
\be \label{BSDElog}
Y_t =  \int_t^T f(s,Y_s)ds + \int_t^T Z_s dW_s
\ee
with driver
\begin{equation}\label{driverlog}
f(t,y) =  \frac{1}{2} \, {\rm dist}^2_t \, \brak{\theta,P}
-\frac{1}{2}|\theta_t|^2- \frac{\alpha y}{h(t)}
- \max_{\tilde{c}\in C} \left(\frac{\alpha}{h} \log(\tilde{c})-\tilde{c}\right)_t
-r - \tilde{e}_t.
\end{equation}
$f(t,y)$ is of linear growth in $y$, and all the other
terms are bounded. It is known from Pardoux and Peng \cite{PP} that \eqref{BSDElog} has a unique solution $(Y,Z)$
such that $Y$ is square-integrable, and it follows from Morlais \cite{Morlais} that $Y$ is bounded and
$Z \in {\cal P}^{1 \times n}_{\rm BMO}$.

\begin{theorem} \label{thmlog}
For $u(x) = \log(x)$, the optimal value of the optimization problem
\eqref{opt} over all admissible strategies is
\be \label{optvaluelog}
h(0)(\log(x)-Y_0),
\ee
and $(\tilde{c}^*, \tilde{p}^*)$ is an optimal admissible strategy if and only if
\be
\label{optstlog}
\tilde{c}^* \in \argmax_{\tilde{c}\in C} \left(\frac{\alpha}{h} \log(\tilde{c})-\tilde{c}\right)
\quad \mbox{and} \quad \tilde{p}^* \in \Pi_{P} \brak{\theta}.
\ee
In particular, an optimal admissible  strategy exists, and it is unique
up to $\nu \otimes \p$-a.e. equality if the sets $C$ and $P$ are ${\cal P}$-convex.
\end{theorem}

\begin{proof}
For every admissible strategy $(\tilde{c},\tilde{p})$, define the process
$$
R^{(c,p)}_t = h(t) e^{- \beta t} \brak{\log \brak{X_t^{(c,p)}}-Y_t}
+\int_0^t \alpha e^{-\beta s} \log(c_s) ds.
$$
Then
$$
R^{(c,p)}_0 = h(0)(\log(x)-Y_0) , \quad R^{(c,p)}_T
= e^{- \beta T}  \log \brak{X_T^{(c,p)}} +\int_0^T \alpha e^{-\beta s} \log(c_s) ds
$$
and
\be \label{diffRlog}
dR^{(c,p)}_t =  h(t)e^{- \beta t}  \edg{( \tilde{p}_t + Z_t) dW_t +  A^{(c,p)}_t dt},
\ee
where
\beas
A^{(c,p)}_t = && \tilde{p}_t\theta_t-\frac{1}{2}|\tilde{p}_t|^2
+ \frac{\alpha Y_t}{h(t)} + f(t,Y_t)
+ \frac{\alpha}{h(t)} \log(\tilde{c}_t) +r + \tilde{e}_t - \tilde{c}_t.
\eeas
First note that
\beas
\tilde{p}_t\theta_t-\frac{1}{2}|\tilde{p}_t|^2+ \frac{\alpha Y_t}{h(t)}
= -\frac{1}{2}|\tilde{p}_t - \theta_t |^2+\frac{1}{2}|\theta_t|^2+ \frac{\alpha Y_t}{h(t)}
\le - \frac{1}{2} \, {\rm dist}^2_t \, \brak{\theta, P}+\frac{1}{2}|\theta_t|^2
+ \frac{\alpha Y_t}{h(t)},
\eeas
and the inequality becomes a $\nu \otimes \p$-a.e. equality if and only if
$$
\tilde{p} \in \Pi_{P} \brak{\theta}.
$$
Furthermore,
$$
\frac{\alpha}{h} \log(\tilde{c}) + r + \tilde{e} - \tilde{c}
\le \max_{\tilde{z} \in C} \left(\frac{\alpha}{h} \log(\tilde{z})-\tilde{z}\right) +r + \tilde{e},
$$
where $\nu \otimes \p$-a.e. equality is attained if and only if
$$
\tilde{c} \in \argmax_{\tilde{z} \in C}
\left(\frac{\alpha}{h} \log(\tilde{z})-\tilde{z}\right).
$$
It follows that for every admissible pair $(\tilde{c},\tilde{p})$,
the process $R^{(c,p)}$ is a local supermartingale. But it can be seen
from \eqref{diffRlog} that the local martingale part of $R^{(c,p)}$ is a true martingale
and its finite variation part is of integrable total variation. So $R^{(c,p)}$
is a supermartingale and one obtains
$$
R^{(c,p)}_0 \ge \E{R^{(c,p)}_T},
$$
where the inequality is strict if $(c,p)$ does not satisfy condition \eqref{optstlog}.
If $(\tilde{c}^*,\tilde{p}^*)$ satisfies \eqref{optstlog}, then the pair is in $C \times P$.
Moreover, we have seen above that it follows from assumption \eqref{asscrra} that the process
$\log(\tilde{c}^*)$ is bounded. The same is true for $\tilde{p}^*$ because
$\theta$ is bounded and, by assumption \eqref{asscrra}, $P$ contains a bounded process $\bar{p}$. In particular,
$(\tilde{c}^*,\tilde{p}^*)$ is admissible and the corresponding process $R^*$ a martingale.
One concludes
$$
R^*_0 = \E{R^*_T},
$$
which shows that $(\tilde{c}^*,\tilde{p}^*)$ is optimal. That a strategy satisfying
\eqref{optstlog} exists follows from \cite{CKV} as well as its uniqueness
(up to $\nu \otimes \p$-a.e. equality) in case $C$ and $P$ are ${\cal P}$-convex.
\end{proof}

\begin{example}
If consumption is unconstrained, that is $C = {\cal P}$, then
$$\tilde{c}^*= \frac{\alpha}{h},\quad
\max_{\tilde{c}\in C} \left(\frac{\alpha}{h}
\log(\tilde{c}) -\tilde{c}\right)= \frac{\alpha}{h}
\brak{\log \brak{\frac{\alpha}{h}}-1},$$
and the driver \eqref{driverlog} becomes
$$
f(t,y) =  \frac{1}{2} {\rm dist}^2_t \brak{\theta,P}
-\frac{1}{2}|\theta_t|^2- \frac{\alpha y}{h(t)}
- \frac{\alpha}{h(t)} \brak{\log \brak{\frac{\alpha}{h(t)}}-1} -r - \tilde{e}_t.
$$
\end{example}

\subsection{Power utility}
\label{subsec:power}

Let us now turn to the case $u(x) = x^{\gamma}/\gamma$ for
$\gamma \in (-\infty,0) \cup (0,1)$. The definition of admissible strategies is
slightly different for $\gamma > 0$ and $\gamma < 0$. But the
optimal value of the optimization problem \eqref{opt}
as well as the optimal strategies will in both cases be of the same form.

\begin{definition}
In the case $\gamma > 0$, an admissible strategy is a pair $(\tilde{c}, \tilde{p})
\in C \times P$ such that
$$
\tilde{c} \ge 0 \quad \nu \otimes \p \mbox{-a.e.} \quad \mbox{and} \quad
\int_0^T \tilde{c}_t dt + \int_0^T |\tilde{p}_t|^2 dt < \infty \quad
\mbox{$\p$-a.s.}
$$
For $\gamma < 0$, we additionally require the process
$(X^{(c,p)})^\gamma$ to be of class {\rm (D)} and
$\E{\int_0^T c^{\gamma}_t dt} < \infty$.
\end{definition}
Note that for $\gamma < 0$, since we assume $x^{\gamma}$ to be $\infty$ if $x \le 0$,
the condition $\E{\int_0^T c^{\gamma}_t dt} < \infty$ implies $\tilde{c} > 0$ $\nu \otimes \p$-a.e.

For every continuous bounded process $Y$, define
\be \label{maxpower}
\max_{\tilde{c}\in C} \left(\frac{\alpha}{\gamma}\tilde{c}^\gamma
e^{Y}-\tilde{c}\right) := \mathop{\rm ess\,sup}_{\tilde{c} \in C}
\left(\frac{\alpha}{\gamma}\tilde{c}^\gamma e^{Y}-\tilde{c}\right),
\ee
where ess\,sup denotes the smallest upper bound with respect to $\nu \otimes \p$-a.e. ordering, and
denote by
\be \label{argmaxpower}
\argmax_{\tilde{c}\in C} \left(\frac{\alpha}{\gamma}\tilde{c}^\gamma
e^{Y}-\tilde{c}\right)
\ee
the set of all processes in $C$ which attain the ess\,sup.
It follows from Cheridito et al. \cite{CKV} that \eqref{argmaxpower} is
not empty and, up to $\nu \otimes \p$-a.e. equality,
contains exactly one process if $C$ is ${\cal P}$-convex.
For the process $\bar{c}$ of assumption \eqref{asscrra}, one has
$$
\frac{\alpha}{\gamma}\bar{c}^\gamma e^{Y}-\bar{c} \le
\max_{\tilde{c}\in C} \left(\frac{\alpha}{\gamma}\tilde{c}^\gamma
e^{Y}-\tilde{c}\right) \le \frac{1-\gamma}{\gamma} \alpha^{1/(1-\gamma)} e^{Y/(1-\gamma)}.
$$
This implies that $\max_{\tilde{c} \in C} \left(\frac{\alpha}{\gamma}\tilde{c}^\gamma
e^{Y}-\tilde{c}\right)$ as well as $u(\tilde{c})$ and $\tilde{c}$ for
every $\tilde{c} \in \argmax_{\tilde{z}\in C} \left(\frac{\alpha}{\gamma}\tilde{z}^\gamma
e^{Y}-\tilde{z}\right)$, are bounded processes. Now consider the BSDE
\be \label{BSDEpower}
Y_t =  \int_t^T f(s,Y_s,Z_s)ds + \int_t^T Z_s dW_s
\ee
with driver
\begin{equation} \label{driverpower}
f(t,y,z)
= \gamma \left(\frac{1-\gamma}{2} \, {\rm dist}^2_t \, \brak{\frac{z+\theta}{1-\gamma}, P}
- \frac{|z+\theta_t|^2}{2(1-\gamma)} - \frac{1}{2\gamma}|z|^2
- \max_{\tilde{c}\in C} \left(\frac{\alpha}{\gamma}\tilde{c}^\gamma
e^{y}-\tilde{c}\right)_t - r - \tilde{e}_t + \frac{\beta}{\gamma}\right).
\end{equation}
Note that $f(t,y,z)$ grows exponentially in $y$.
But it satisfies Assumption (A.1) in Briand and Hu \cite{BH2}.
So it can be deduced from Proposition 3 in \cite{BH2} that \eqref{BSDEpower} has a solution
$(Y,Z)$ such that $Y$ is bounded. That $Z$ is in ${\cal P}^{1 \times n}_{\rm BMO}$
and the uniqueness of such a solution then follow from \cite{Morlais}.

\begin{theorem} \label{thmpower}
If $u(x) = x^{\gamma}/\gamma$ for $\gamma \in (-\infty,0) \cup (0,1)$, the
optimal value of the optimization problem \eqref{opt}
over all admissible strategies is
\be \label{optvaluepower}
\frac{1}{\gamma} x^{\gamma} e^{-Y_0},
\ee
and $(\tilde{c}^*, \tilde{p}^*)$ is an optimal admissible strategy if and only if
\be
\label{optstpower}
\tilde{c}^* \in \argmax_{\tilde{c} \in C}
\left(\frac{\alpha}{\gamma}\tilde{c}^\gamma e^{Y}-\tilde{c}\right)
\quad \mbox{and} \quad \tilde{p}^* \in \Pi_{P} \brak{\frac{Z + \theta}{1-\gamma}}.
\ee
In particular, an optimal admissible strategy exists, and it is unique
up to $\nu \otimes \p$-a.e. equality if the sets $C$ and $P$ are ${\cal P}$-convex.
\end{theorem}

\begin{proof}
For every admissible strategy $(\tilde{c},\tilde{p})$ define the process
$$
R^{(c,p)}_t = e^{- \beta t} \frac{1}{\gamma}
\brak{X_t^{(c,p)}}^\gamma e^{-Y_t}
+ \int_0^t \alpha e^{-\beta s} \frac{1}{\gamma}c_s^\gamma ds.
$$
Then
$$
R^{(c,p)}_0 = \frac{1}{\gamma} x^\gamma e^{-Y_0},
\quad R^{(c,p)}_T = e^{- \beta T} \frac{1}{\gamma}
\brak{X_T^{(c,p)}}^\gamma
+ \int_0^T \alpha e^{-\beta s}\frac{1}{\gamma}c_s^\gamma ds
$$
and
$$
dR^{(c,p)}_t =  e^{- \beta t}
\brak{X_t^{(c,p)}}^\gamma e^{-Y_t} \edg{\brak{\tilde{p}_t
+\frac{1}{\gamma}Z_t} dW_t +  A^{(c,p)}_t dt},
$$
where
\beas
A^{(c,p)}_t = && \tilde{p}_t (Z_t+\theta_t)+\frac{1}{2}(\gamma-1) |\tilde{p}_t|^2
+\frac{1}{2\gamma}|Z_t|^2+\frac{1}{\gamma} f(t,Y_t,Z_t)\\
&& + \frac{\alpha}{\gamma}\tilde{c}_t^\gamma e^{Y_t}
+ \tilde{e}_t - \tilde{c}_t + r - \frac{\beta}{\gamma}.
\eeas
First note that
\beas
& & \tilde{p}_t (Z_t+\theta_t) + \frac{1}{2}(\gamma-1) |\tilde{p}_t|^2 + \frac{1}{2\gamma}|Z_t|^2\\
&=& \frac{\gamma-1}{2} \left|\tilde{p}_t-\frac{Z_t + \theta_t}{1-\gamma} \right|^2
+\frac{1}{2(1-\gamma)}|Z_t+\theta_t|^2+\frac{1}{2\gamma}|Z_t|^2\\
&\le& \frac{\gamma-1}{2} {\rm dist}^2_t \brak{\frac{Z + \theta}{1-\gamma}, P}
+ \frac{1}{2(1-\gamma)}|Z_t+\theta_t|^2+\frac{1}{2\gamma}|Z_t|^2,
\eeas
and the inequality becomes a $\nu \otimes \p$-a.e. equality if and only if
$$
\tilde{p} \in \Pi_{P} \brak{\frac{Z + \theta}{1-\gamma}}.
$$
Furthermore,
$$
\frac{\alpha}{\gamma} \tilde{c}^\gamma e^{Y} + \tilde{e}
-\tilde{c} + r - \frac{\beta}{\gamma}\\
\le \max_{\tilde{z} \in C} \left(\frac{\alpha}{\gamma}\tilde{z}^\gamma e^{Y}
-\tilde{z}\right) + \tilde{e} +r - \frac{\beta}{\gamma},
$$
where $\nu \otimes \p$-a.e. equality is attained if and only if
$$
\tilde{c}\in \argmax_{\tilde{z} \in C} \left(\frac{\alpha}{\gamma}\tilde{z}^\gamma
e^{Y}-\tilde{z}\right).
$$
The next step of the proof is slightly different for the two cases $\gamma > 0$ and $\gamma< 0$.
Let us first assume $\gamma > 0$. Then for every admissible pair $(\tilde{c},\tilde{p})$,
the process $R^{(c,p)}$ is a positive local supermartingale,
and hence, a supermartingale. In particular,
\be \label{supg>0}
R^{(c,p)}_0 \ge \E{R^{(c,p)}_T}
\ee
with strict inequality if condition \eqref{optstpower} is violated.
Now let $(\tilde{c}^*,\tilde{p}^*)$ be a strategy satisfying \eqref{optstpower}.
Then, as we have seen above, $u(\tilde{c}^*)$ and $\tilde{c}^*$ are bounded processes.
In particular $\tilde{c}^* \ge 0$ $\nu \otimes \p$-a.e.
Moreover, since $\theta$ and $\bar{p}$ are bounded and $Z$ is in ${\cal P}^{1 \times n}_{\rm BMO}$,
it follows from
$$
\mbox{dist}\brak{\frac{Z + \theta}{1-\gamma}, \tilde{p}^*}
= \mbox{dist}\brak{\frac{Z + \theta}{1-\gamma}, P} \le \mbox{dist}
\brak{\frac{Z + \theta}{1-\gamma}, \bar{p}},
$$
that
$\mbox{dist} \brak{\frac{Z + \theta}{1-\gamma}, \bar{p}}$
is in ${\cal P}_{\rm BMO}$ and $\tilde{p}^*$ in ${\cal P}^{1 \times n}_{\rm BMO}$.
This shows that $(\tilde{c}^*,\tilde{p}^*)$ is admissible.
Furthermore, the corresponding process $X^*$ satisfies
\be \label{Xexp}
X^*_t = x \, {\cal E} \brak{\tilde{p}^* \cdot W^{\q}}_t \exp\brak{\int_0^t( r
+ \tilde{e}_s - \tilde{c}^*_s) ds} \le M {\cal E} \brak{\tilde{p}^* \cdot W^{\q}}_t
\ee
for some constant $M \in \mathbb{R}_+$. Choose $\gamma < \gamma' <1$ and
let $\varepsilon = 1 - \gamma'$. Since $\theta$ is bounded, one has
$\EQ{{\cal E}\brak{\theta \cdot W^{\q}}_T^{1/\varepsilon}} < \infty$,
and by H\"older's inequality, one obtains for every stopping time $\tau \le T$,
\beas
&& \E{(X^*_{\tau})^{\gamma'}}
= \EQ{(X^*_{\tau})^{\gamma'} {\cal E}(\theta \cdot W^{\q})_T}\\
&\le& \EQ{(X^*_{\tau})}^{\gamma'}
\EQ{{\cal E} \brak{\theta \cdot W^{\q}}_T^{1/\varepsilon}}^{\varepsilon}
\le M^{\gamma'} \EQ{{\cal E}\brak{\theta \cdot W^{\q}}_T^{1/\varepsilon}}^{\varepsilon}.
\eeas
It follows that $(X^*)^{\gamma}$ is of class (D) and $R^*$ a martingale.
In particular,
$$
R^*_0 = \E{R^*_T},
$$
which shows that $(\tilde{c}^*, \tilde{p}^*)$ is optimal.

If $\gamma < 0$, $R^{(c,p)}$ is for every admissible pair $(\tilde{c},\tilde{p})$
a supermartingale due to our assumption that $(X^{(c,p)})^{\gamma}$ is of class (D)
and $\E{\int_0^T c^{\gamma}_t dt} < \infty$. So again,
$$
R^{(c,p)}_0 \ge \E{R^{(c,p)}_T}
$$
with strict inequality if $(\tilde{c},\tilde{p})$ does not fulfill condition \eqref{optstpower}.
If $(\tilde{c}^*,\tilde{p}^*)$ satisfies \eqref{optstpower}, it follows as in the case $\gamma > 0$, that
$u(\tilde{c}^*)$ and $\tilde{c}^*$ are bounded and
$\tilde{p}^*$ belongs to ${\cal P}^{1 \times n}_{\rm BMO}$.
In particular, $\tilde{c}^* > 0$ $\nu \otimes \p$-a.e.
Moreover, $-R^*$ is a positive local martingale. So
$-R^*$ is a supermartingale and $\mathbb E[-R_T^*] < \infty$.
Hence,
\be \label{X*est}
\E{(X^*_T)^\gamma + \int_0^T (c^*_t)^{\gamma} dt} < \infty,
\ee
and it follows that $(\tilde{c}^*, \tilde{p}^*)$ is admissible.
It can be seen from \eqref{posX} that \eqref{X*est} implies
$$
\E{{\cal E} (\tilde{p}^* \cdot W^{\q})_T^{\gamma}} < \infty,
$$
where $W^{\q}_t = W_t + \int_0^t \theta_s ds$.
So one obtains from Jensen's inequality that for every stopping time $\tau \le T$,
\beas
&& {\cal E}(\tilde{p}^* \cdot W^{\mathbb Q})_{\tau}^{\gamma}
\le \brak{\EQ{{\cal E}(\tilde{p}^*\cdot
W^{\mathbb Q})_T^{\gamma/2}|{\cal F}_{\tau}}}^2\\
&=& \brak{\E{{\cal E}(\tilde{p}^*\cdot W^{\mathbb Q})_T^{\gamma/2}
\frac{{\cal E}(-\theta\cdot W)_T}{{\cal E}(-\theta\cdot W)_{\tau}}|{\cal F}_\tau}}^2\\
&\le& \E{{\cal E}(\tilde{p}^*\cdot W^{\mathbb Q})_T^{{\gamma}}|{\cal F}_{\tau}}
\E{\frac{{\cal E}(-\theta\cdot W)^2_T}{{\cal E}(-\theta \cdot W)_{\tau}^2}|{\cal F}_{\tau}}\\
&\le& M \E{{\cal E}(\tilde{p}^*\cdot W^{\mathbb Q})_T^{{\gamma}} |{\cal F}_{\tau}}
\eeas
for some constant $M \in \mathbb{R}_+$. This shows that
${\cal E} (\tilde{p}^* \cdot W^{\q})^{\gamma}$ and $(X^*)^\gamma$ are of class (D).
Hence, $(\tilde{c}^*, \tilde{p}^*)$ is admissible and $R^*$ is a martingale.
In particular, $R^*_0 = \E{R^*_T}$, and it follows that
$(\tilde{c}^*,\tilde{p}^*)$ is optimal.

In both cases, $\gamma > 0$ and $\gamma < 0$, existence of an optimal strategy
follows from \cite{CKV} as does uniqueness (up to $\nu \otimes \p$-a.e. equality)
if the sets $C$ and $P$ are ${\cal P}$-convex.
\end{proof}

\begin{example}
If consumption is unconstrained, that is $C = {\cal P}$, then
$$c^* = \alpha^{1/(1-\gamma)} e^{Y_t/(1-\gamma)}, \quad
\max_{\tilde{c}\in C} \left(\frac{\alpha}{\gamma}\tilde{c}^\gamma
e^{y}-\tilde{c}\right)
=\frac{1-\gamma}{\gamma} \alpha^{1/(1-\gamma)} e^{y/(1-\gamma)},$$
and the driver \eqref{driverpower} becomes
$$
f(t,y,z)
= \gamma\left(\frac{1-\gamma}{2} {\rm dist}^2_t \brak{ \frac{z+\theta}{1-\gamma}, P}
- \frac{|z+\theta_t|^2}{2(1-\gamma)} - \frac{1}{2\gamma}|z|^2
- \frac{1-\gamma}{\gamma} \alpha^{1/(1-\gamma)} e^{y/(1-\gamma)}
-r - \tilde{e}_t + \frac{\beta}{\gamma}\right).
$$
\end{example}

\setcounter{equation}{0}
\section{Conclusion}
\label{sec:conclusion}

We gave solutions to optimal consumption and investment problems for expected
utility optimizers in the three cases of exponential, logarithmic and power
utility. In the exponential case we assumed that
the income rate $e$ as well as the final payment $E$ were bounded and consumption
$c$ was unconstrained. The proof of Theorem \ref{thmexp} relies on the fact that the
BSDE \eqref{BSDEexp} has a unique solution $(Y,Z)$ such that $Y$ is bounded and
$Z$ is in ${\cal P}^{1 \times n}_{\rm BMO}$.
There exist extensions of the result of Kobylanski \cite{Kobylanski}
showing that equation \eqref{BSDEexp} also has a unique solution for certain unbounded
random variables $E$; see Briand and Hu \cite{BH1,BH2}, Ankirchner et al.
\cite{AIP}, Delbaen et al. \cite{DH2}. However, if $E$ is not bounded, $Y$ is not bounded
and $Z$ not necessarily in ${\cal P}^{1 \times n}_{\rm BMO}$. It is still possible to show that there exist
strategies satisfying condition \eqref{optstexp}. But one would need a new argument to show
that and in which sense they are optimal.
If one introduces restrictions on consumption, the proof of Theorem \ref{thmexp} does not go through.
If one can show existence of an optimal strategy when consumption is
constrained, it obviously has to be different from \eqref{optstexp}.
In the cases of logarithmic and power utility we assumed $\tilde{e} = e/X$ to be bounded
and $E = 0$. The first assumption is technical and ensures that the BSDEs \eqref{BSDElog} and
\eqref{BSDEpower} both have unique solutions $(Y,Z)$ such that $Y$ is bounded and $Z$ in
${\cal P}^{1 \times n}_{\rm BMO}$. Again, if \eqref{BSDElog} or \eqref{BSDEpower} can
be solved for unbounded $\tilde{e}$, one can still show that there exist strategies satisfying
\eqref{optstlog} or \eqref{optstpower}, respectively. But then again, one would have to find
a new explanation why and in which sense they are optimal. For $E \neq 0$ the proofs of
Theorems \ref{thmlog} and \ref{thmpower} do not work because the process
$R^{(c,p)}$ does not have the correct terminal value. One would have to find a new way
to construct $R^{(c,p)}$ to cover this case. We point out that also in Sections 3 and 4 of
Hu et al. \cite{HIM} as well as in Nutz \cite{Nutz} it is assumed that $E = 0$.

\end{document}